\documentclass[aps,pra,reprint,unsortedaddress,superscriptaddress]{revtex4-2}

\usepackage{times}
\usepackage{amsmath, amsfonts, amssymb, amsthm}
\usepackage{braket}     
\usepackage{graphicx}
\usepackage[psdextra]{hyperref} 
\hypersetup{
	colorlinks=true,
	linkcolor=red,      
	anchorcolor=blue,
	citecolor=blue,    
	urlcolor=red        
}
\newtheorem{lemma}{Lemma}
\newtheorem{theorem}{Theorem}
\newtheorem{corollary}[theorem]{Corollary}

\theoremstyle{definition}

\newcommand{\casql}{Laboratory of Quantum Information, University of Science and Technology of China, Hefei, Anhui, 230026, China}

\newcommand{\aihf}{Institute of Artificial Intelligence, Hefei Comprehensive National Science Center, Hefei, Anhui, 230088, China}
\newcommand{\origin}{Origin Quantum Computing Technology (Hefei) Co., Ltd., Hefei, Anhui, 230026, China}

\newcommand{\norm}[1]{\left\lVert#1\right\rVert}
\newcommand{\abs}[1]{\left\lvert#1\right\rvert}
\newcommand{\Ocal}{\mathcal{O}}
\newcommand{\tOcal}{\tilde{\mathcal{O}}}

\begin{document}
	
	\title{A Unified Poisson Summation Framework for Generalized Quantum Matrix Transformations}
	
	\author{Chao Wang}
	\affiliation{\origin}
	\author{Xi-Ning Zhuang}
	\affiliation{\casql}  
	\affiliation{\origin}
	\author{Menghan Dou}
	\affiliation{\origin}
	\author{Zhao-Yun Chen}
	\email{chenzhaoyun@iai.ustc.edu.cn}
	\affiliation{\aihf}
	\author{Guo-Ping Guo}    
	\affiliation{\origin}
	\affiliation{\aihf}
	\affiliation{\casql}
	
	\date{\today}
	
	\begin{abstract}
		We present a unified algorithmic framework for quantum simulation of non-unitary dynamics and matrix functions, governed by the principle of spectral aliasing derived from the Poisson Summation Formula (PSF). By reinterpreting discretization errors as spectral folding in dual domains, we synthesize two distinct algorithmic paths: (i) the Fourier-PSF path, generalizing transmutation methods for time-domain filtering, which is optimal for singular and fractional dynamics $e^{-tH^\alpha}$, here $H\succeq 0$; and (ii) the contour-PSF path, a novel discrete contour transform based on the resolvent formalism, which achieves exponential convergence for holomorphic matrix functions via radius optimization. This dual framework resolves the smoothness-sparsity trade-off: it utilizes the Fourier basis to handle branch-point singularities where analyticity fails, and the Resolvent basis to exploit complex-plane regularity where it exists. We demonstrate the versatility of this framework by efficiently simulating diverse phenomena, from fractional anomalous diffusion to high-precision solutions of stiff differential equations, outperforming existing methods in their respective optimal regimes.
	\end{abstract}
	
	\maketitle
	\section{Introduction}
	
	The advent of quantum computing represents a paradigm shift in both physics and computer science, unlocking computational frontiers that remain inaccessible to classical approaches. This advantage is most notably demonstrated in specific algorithmic breakthroughs, including integer factorization~\cite{doi:10.1137/S0036144598347011}, unstructured search optimization~\cite{10.1145/237814.237866}, and the solution of linear systems~\cite{PhysRevLett.103.150502, doi:10.1137/16M1087072}, alongside the simulation of quantum many-body dynamics~\cite{RevModPhys.86.153, Altman2021qsim, daley2022practical}. Fundamentally, these accelerations stem from the exponential representational capacity of the Hilbert space, which grants quantum computers a distinct advantage in performing large-scale matrix transformations.
	
	Consequently, a wide array of quantum applications ranging from Hamiltonian simulation and quantum linear solvers to quantum walks, which can be unified under the mathematical framework of evaluating matrix functions $f(A)$~\cite{doi:10.1137/16M1087072,1366222,doi:10.1073/pnas.1801723115,doi:10.1126/science.273.5278.1073, 20121122333}. While the quantum singular value transformation (QSVT)~\cite{10.1145/3313276.3316366} has emerged as a powerful unified framework for such tasks, often achieving exponential speedups, its direct applicability to general matrix functions $f(A)$ is constrained. Standard QSVT typically requires the input matrix $A$ to admit a trivial singular value decomposition (e.g., normal or Hermitian matrices) or restricts $f$ to specific forms such as matrix inversion.
	
	The challenge becomes particularly pronounced in the realm of physical simulation. While Hamiltonian simulation for Hermitian operators $H$ is a cornerstone of quantum algorithms~\cite{PhysRevLett.103.150502,doi:10.1137/16M1087072,Berry2007,doi:10.1126/science.273.5278.1073, PhysRevLett.114.090502, farhi2014quantumapproximateoptimizationalgorithm, Lloyd2014}, a vast class of problems necessitates the simulation of non-Hermitian dynamics $e^{-At}$. These range from the heat equation~\cite{Linden2022} and Fokker-Planck dynamics~\cite{pavliotis2014stochastic, Jin_2023} to financial modeling via the Black-Scholes equation~\cite{Gonzalez_Conde_2023} and open quantum systems~\cite{RevModPhys.93.015008, Liu2025simulationofopen, delgado2025quantum}. In these contexts, where typically $\Re(A):=\frac{A+A^\dagger}{2} \succeq 0$, early approaches explored spectral methods~\cite{childs2020spectral}, truncated Dyson series~\cite{Berry_2024}, and time-marching schemes~\cite{Fang2023timemarchingbased}. To further optimize query complexity, subsequent methods embedded these non-Hermitian problems into larger Hermitian frameworks, using techniques such as Schr\"odingerization~\cite{PhysRevA.108.032603, JIN2022111641, jin2025schrodingerizationmethodlinearnonunitary, PhysRevLett.133.230602}, linear combination of Hamiltonian simulations (LCHS)~\cite{PhysRevLett.131.150603,an2023quantumalgorithmlinearnonunitary,low2025optimalquantumsimulationlinear}, and contour-based matrix decomposition (CBMD)~\cite{wang2025quantumsimulationnonunitarydynamics}.

	Simultaneously, to address the more general problem of matrix function transformation $f(A)$, algorithms such as quantum eigenvalue transformation (QET) and methods defined via contour integrals were introduced~\cite{10756112,Takahira_2020,takahira2021quantumalgorithmsbasedblockencoding}. Subsequent advancements have demonstrated that algorithms utilizing contour integration can achieve superior query complexity~\cite{jiang2026contourintegralbasedquantumeigenvalue}. To realize these optimizations, these frameworks typically integrate foundational quantum primitives, such as quantum signal processing (QSP)~\cite{PhysRevLett.118.010501}, QSVT~\cite{10.1145/3313276.3316366}, and linear combination of unitaries (LCU)~\cite{20121122333}, with rigorous mathematical transformations.
	
	While these pioneering algorithms have significantly advanced the field, their specific mathematical formulations naturally define their optimal application scopes. For instance, QET provides a robust framework for general functions by approximating them with Chebyshev polynomials~\cite{10756112}; however, this approach introduces an explicit query complexity dependence of $\mathcal{O}(d \log(1/\epsilon))$ on the polynomial degree $d$, which can become computationally intensive for certain high order transformations~\cite{10756112, Low2026quantumlinearsystem}. Similarly, approaches relying on contour integrals practically necessitate the discretization of continuous integrals, a process that can lead to overestimations of approximation errors and inaccurate evaluations of the required number of LCU terms~\cite{Takahira_2020, takahira2021quantumalgorithmsbasedblockencoding, jiang2026contourintegralbasedquantumeigenvalue}. Furthermore, methods involving the Laplace transform (Lap-LCHS) offer an elegant solution for functions possessing an integrable inverse Laplace transform, though this condition naturally places fundamental cases such as matrix polynomials outside their primary scope~\cite{an2024laplacetransformbasedquantum}. Finally, even highly accurate decomposition strategies like the contour based matrix decomposition (CBMD) are primarily optimized for scenarios where the modulus of the transform function $|f(z)|$ grows at a controlled rate along the real or purely imaginary axes~\cite{wang2025quantumsimulationnonunitarydynamics}. Consequently, as these structural and analytic conditions delineate the current boundaries of methodologies, developing a generalized query complexity analysis that universally accommodates broader applications remains a vital frontier in the design of quantum algorithms.
	
	In this work, we propose a generalized framework that transcends these limitations by identifying the fundamental source of discretization error: \textit{Spectral Aliasing}. By viewing the simulation problem through the lens of harmonic analysis, specifically the PSF~\cite{Miller2004}, we synthesize a dual path framework that operates without relying on structural assumptions like auxiliary differential equations. Specifically, we unify the following two distinct algorithmic paths:
	
	The first path (the Fourier-PSF approach) is tailored for structured non-Hermitian simulation $e^{-At}$ where $A=H^\alpha$ and $H\succeq0$. This generalizes transmutation methods (such as the Kannai transform for $\alpha=2$~\cite{Kannai01011977,jin2026transmutationbasedquantumsimulation}) to arbitrary $\alpha\geq 0.5$. This path is optimal for singular and fractional dynamics ($\alpha \notin \mathbb{Z}$), where standard contour integration fails due to branch points. By mapping the dynamics to an LCU weighted by Fourier coefficients, we efficiently handle heavy tailed distributions (e.g., Lévy flights).
	
	The second path (the contour-PSF approach) introduces a novel discrete contour transform for general holomorphic matrix functions $f(A)$. Crucially, we reinterpret the discretized contour integral not as a quadrature approximation, but as a finite PSF over the cyclic group of roots of unity $\mathbb{Z}_m$. This group theoretic perspective rigorously identifies the approximation error as geometric aliasing in the dual lattice (matrix powers $A^{km}$). This proves that for functions analytic merely on the domain covering the eigenvalues, such as $e^{-At}$ or rational functions, the convergence is exponential. This allows us to achieve high precision with a logarithmically small number of LCU nodes $\tOcal(\log(1/\epsilon))$, contrasting sharply with standard numerical analyses that suggest polynomial costs~\cite{ jiang2026contourintegralbasedquantumeigenvalue}.
	
	By establishing these two complementary paths, our framework reveals that the choice of discretization scheme, whether in the time domain or the complex plane, is strictly dual to the choice of the basis set best suited for the analytic properties of the function. Crucially, our method requires $f(z)$ to be holomorphic only within a neighborhood of the spectrum, thereby relaxing the strict global analyticity assumptions often found in prior works. Collectively, our analysis uncovers a fundamental smoothness sparsity tradeoff: the regularity of the target function in the dual domain strictly dictates the sparsity of the LCU decomposition in the primary domain. This unified framework not only simplifies the understanding of existing methods but also provides a rigorous guide for selecting the quantum algorithm to simulate diverse physical phenomena efficiently.

	\section{Preliminaries and General Theory}
	\label{sec:preliminaries}
	
	\subsection{Notation and Quantum Primitives}
	\label{subsec:primitives}
	
	Our algorithm relies on the standard oracle model for accessing the generator and employs the LCU~\cite{20121122333} and QSVT~\cite{10.1145/3313276.3316366} as the central algorithmic backbone.
	
	\subsubsection{Block-Encoding and Matrix Arithmetics}
	We assume access to the generator via the block-encoding framework. To avoid confusion with the decay order $\alpha$ used throughout this paper, we denote the factor of the block-encoding as $\alpha_H\geq \norm{H}$. Formally, a unitary $U_H$ is a $(\alpha_H, m, 0)$-block-encoding of a Hermitian matrix $H$ if it acts on $m$ ancilla qubits such that:
	\begin{equation}
		H = \alpha_H (\langle 0 |^{\otimes m} \otimes I) U_H (| 0 \rangle^{\otimes m} \otimes I).
	\end{equation}
	Efficient constructions of such oracles for sparse or structured matrices are well-established using fast approximate quantum circuits~\cite{Camps_2022, Sunderhauf2024blockencoding}. 
	
	Leveraging the QSVT or QSP framework~\cite{PhysRevLett.118.010501, 10.1145/3313276.3316366}, the time-evolution operator $\cos(Ht)$ can be implemented with optimal query complexity. Specifically, simulating $H$ for time $t$ with precision $\epsilon$ requires $\mathcal{O}(\|H\|t + \log(1/\epsilon))$ queries to $U_H$~\cite{20121122333}. In our framework, these unitary evolutions constitute the basis operations in the LCU decomposition. 
	
	In the context of approximate matrix inversion schemes, for a matrix $A$ satisfying $\|A\| \leq 1$ and $\|A^{-1}\| \leq 1/\delta$, there exists a polynomial of degree $\Ocal\left(\frac{1}{\delta}\log\frac{1}{\epsilon}\right)$ that approximates $\frac{3\delta}{4}A^{-1}$ with a spectral norm error bounded by $\epsilon$~\cite{10.1145/3313276.3316366}.
	
	\subsubsection{Linear Combination of Unitaries }
	To realize the non-unitary operator $V \approx \sum_{k} c_k U_k$, where $c_k \in \mathbb{C}$ are complex coefficients and $U_k$ are unitary operators, we employ the LCU technique~\cite{20121122333}. This method requires constructing two oracle operators:
	
	State preparation ($O_{c,l}$ and $ O_{c,r}$): Encodes the coefficients into the amplitude of an ancilla state.
	\begin{equation}
		\begin{split}
			O_{c,l} |0\rangle = \frac{1}{\sqrt{\|{c}\|_1}} \sum_{k}\overline{ \sqrt{c_k} }|k\rangle,\\
			O_{c,r} |0\rangle = \frac{1}{\sqrt{\|{c}\|_1}} \sum_{k}{ \sqrt{c_k} }|k\rangle,
		\end{split}
	\end{equation}
	where $\|{c}\|_1 = \sum_k |c_k|$ is the $L^1$-norm of the coefficients.
	
	Select oracle ($\mathrm{SEL}$): Applies the target unitary conditional on the ancilla register.
	\begin{equation}
		\mathrm{SEL} := \sum_{k} |k\rangle\langle k| \otimes U_k.
	\end{equation}
	
	The operation $(O_{c,l}^\dagger \otimes I)\mathrm{SEL}(O_{c,r}\otimes I)\ket{u_{0}}$ implements the target linear combination upon measuring $|0\rangle$ in the ancilla register:
	\begin{equation}
		\frac{1}{\|c\|_1}\ket{0} \sum_k c_k U_k |u_{0}\rangle +\ket{\perp}.\nonumber
	\end{equation}
	The success probability of this post-selection is about $ (\|V\ket{u_{0}}\| / \|c\|_1)^2$. Since non-unitary dynamics often imply $\|V\ket{u_{0}}\|< 1$, the total complexity is scaled by the overhead of oblivious amplitude amplification~\cite{PhysRevLett.114.090502}, which introduces a factor of $\mathcal{O}(\|c\|_1 / \|V\ket{u_{0}}\|)$ to the query count. Therefore, minimizing the coefficient 1-norm $\|c\|_1$ (sparsity) is critical for algorithmic efficiency.

	\subsection{Problem Formulation}
	\label{subsec:problem_formulation}
	Our primary objective is to simulate the time evolution of a quantum state $u_T:=u(T)$ governed by a generalized dissipative or dispersive equation. Specifically, we aim to implement the non-unitary evolution operator corresponding to the $\alpha$-th order decay:
	\begin{equation}\label{eq:origin_problem}
		u(T) = e^{-T {H}^{\alpha}} u_0= e^{-T A} u_0, \quad \alpha > 0,
	\end{equation}
	where $H\succeq 0$ is a Hermitian operator. This formulation encompasses the standard heat equation ($\alpha=2$), biharmonic diffusion ($\alpha=4$), and fractional anomalous diffusion ($\alpha \notin \mathbb{Z}$ and $\alpha \geq 0.5$).
	
	We assume access to the generator through a root operator $H$, which is a Hermitian matrix satisfying $H^\alpha = A$. Practically, we assume access to a block-encoding $U_H$ of $H$.
	
	More generally, our framework extends to the simulation of a broad class of matrix functions beyond the specific power-law decay. We consider the transformation of an initial quantum state $|\psi\rangle$ under a general matrix function $f(A)$. Formally, we consider a normal matrix $A \in \mathbb{C}^{N \times N}$, accessible via its unitary block-encoding oracle $U_A$. Given a scalar function $f: \mathcal{D} \to \mathbb{C}$, where $\mathcal{D} \subset \mathbb{C}$ is a holomorphic domain containing the spectrum $\sigma(A)$, our goal is to prepare an $\epsilon$-approximate normalized target state:
	\begin{equation}
		\frac{f(A)|\psi\rangle}{\|f(A)|\psi\rangle\|}
	\end{equation}
	with success probability $\Omega(1)$. This formulation provides a unified framework for diverse quantum simulation tasks: explicitly, setting $f(z) = e^{-izt}$ recovers standard Hamiltonian simulation, while $f(z) = e^{-zt}$ (with $\Re(z) \ge 0$) corresponds to generalized dissipative dynamics. Furthermore, this setting extends to broad classes of matrix arithmetic operations where $f(z)$ represents rational functions or fractional powers. The algorithmic challenge lies in constructing a quantum circuit that achieves this mapping with optimal query complexity to $U_A$, strictly determined by the analytic properties of $f$ on the spectral domain $\sigma(A)$.

	\section{Mathematical Foundation: A Unified Poisson Summation Framework}
	\label{sec:math_foundation}
	
	The theoretical cornerstone of our framework is the Poisson summation principle, which rigorously connects the discrete sampling of a function to the spectral repetitions of its transform~\cite{stein1971introduction}. We demonstrate that this principle governs the discretization error in two complementary regimes: the continuous time-frequency domain via the standard PSF, and the complex spectral domain via a novel finite Poisson summation over cyclic groups.
	
	\subsection{The Continuous Regime: Fourier-PSF for Hamiltonian Simulation}\label{sec:fpsf}
	
	Let $f(x)$ be a function in the Schwartz space or with sufficient decay~\cite{stein1971introduction}, and let its Fourier transform be defined as $\hat{f}(\xi) = \int_{\mathbb{R}} f(x) e^{-2\pi i \xi x} dx$. The generalized PSF, tailored for our discretization scheme with a scaling parameter $a > 0$ and shift $\delta \in \mathbb{R}$, states:
	\begin{equation} \label{eq:poisson_scalar_scaled}
		\frac{1}{a} \sum_{k \in \mathbb{Z}} f(k/a) e^{-i 2\pi k \delta/a} = \sum_{n \in \mathbb{Z}} \hat{f}(an + \delta).
	\end{equation}
	This identity establishes a fundamental duality: the sampling rate in the time domain (controlled by $a$) dictates the separation of spectral copies in the frequency domain.
	
	Invoking the spectral mapping theorem~\cite{PhysRevLett.131.150603}, we promote the scalar argument $\delta$ to a Hermitian operator ${H}\succeq 0$. This yields the operator-valued identity that forms the basis of our LCU approach:
	\begin{equation} \label{eq:poisson_operator}
		\begin{split}
			&\underbrace{\frac{1}{a} \sum_{|k| \leq K} f(k/a) e^{-i \frac{2\pi k}{a} {H}}}_{\text{LCU Implementation}} + \underbrace{\frac{1}{a} \sum_{|k| > K} f(k/a) e^{-i \frac{2\pi k}{a} {H}}}_{\text{Truncation Error}} 
			\\
			&=\underbrace{\hat{f}({H})}_{\text{Target}} + \underbrace{\sum_{n \neq 0} \hat{f}({H} + naI)}_{\text{Aliasing Error}}.
		\end{split}
	\end{equation}
	
	As will be seen later, in this paper we only consider the case where $f$ and $\hat{f}$ are even. Consequently, the above equation simplifies to:
	\begin{equation} \label{eq:poisson_operator_2}
		\begin{split}
			&\underbrace{\frac{1}{a} \sum_{|k| \leq K} f(k/a) \cos (
				\frac{2\pi k}{a}\sqrt{H})}_{\text{LCU Implementation}} + \underbrace{\frac{1}{a} \sum_{|k| > K} f(k/a) e^{-i \frac{2\pi k}{a} \sqrt{H}}}_{\text{Truncation Error}} 
			\\
			&=\underbrace{\hat{f}(\sqrt{H})}_{\text{Target}} + \underbrace{\sum_{n \neq 0} \hat{f}(\sqrt{H} + naI)}_{\text{Aliasing Error}}.
		\end{split}
	\end{equation}

	The significance of Eq.~\eqref{eq:poisson_operator_2} is twofold. The left-hand side represents the implementable LCU, where the constituent operators $\cos\left(\frac{2\pi k}{a}\sqrt{H}\right)$ are weighted by coefficients $c_k = f(k/a)/a$. Although $\cos\left(\frac{2\pi k}{a}\sqrt{H}\right)$ is an entire function, directly implementing it via QSVT using the block-encoding of $\tilde{H} = H/\|H\|$ incurs a prohibitive overhead. This limitation arises because the corresponding polynomial is not bounded by a constant across the entire required working domain $x \in [-1, 1]$. Specifically, for $x < 0$, the function analytically continues to a hyperbolic cosine, $\cosh\left(\frac{2\pi k}{a}\sqrt{\|H\||x|}\right)$, which grows exponentially and violently violates the QSVT norm constraint bounded by $1$.
	
	To circumvent this norm-bounding issue and construct a rigorously QSVT-compliant polynomial, we employ a domain-shifting technique. By substituting the input model, we can evaluate the modified polynomial $\cos\left(\frac{2\pi k \sqrt{\|H\|}}{a}\sqrt{\frac{x+1}{2}}\right)$, which necessitates the block-encoding of $2\tilde{H}-I$. Alternatively, one can directly query the block-encoding of $\sqrt{\tilde{H}}$ to process the polynomial $\cos\left(\frac{2\pi k\sqrt{\|H\|}}{a} x\right)$. Crucially, both reformulated polynomials map the operational domain strictly into $[-1, 1]$ and maintain an absolute value bounded by $1$ for all $x \in [-1, 1]$, making them inherently suitable for the QSVT framework. 
	
	Furthermore, constructing the block-encoding of $2\tilde{H}-I$ from $\sqrt{\tilde{H}}$ requires only two queries via a quadratic polynomial transformation (i.e., the Chebyshev polynomial $2x^2-1$). Therefore, assuming access to $2\tilde{H}-I$ is computationally equivalent and not strictly more demanding than assuming access to $\sqrt{\tilde{H}}$. This structural adjustment allows us to formalize the oracle query complexity for each LCU component as follows:
	
	\begin{lemma}\label{lemma:qsvt_complexity}
		Let $\tilde{H} = H/\|H\|\succeq 0$ be the normalized Hamiltonian with its spectrum bounded in $[0, 1]$. Assuming access to the block-encoding oracle of either $2\tilde{H}-I$ or $\sqrt{\tilde{H}}$, implementing the target operator $\cos\left(\frac{2\pi k}{a}\sqrt{H}\right)$ to precision $\epsilon$ via QSVT requires a query complexity of:
		\begin{equation}
			\mathcal{O}\left( \frac{k}{a}\sqrt{\|H\|} + \log\left(\frac{1}{\epsilon}\right) \right)
		\end{equation}
		calls to the respective input oracle.
	\end{lemma}
	\begin{proof}
		Let the scaled phase factor be $\tau = \frac{2\pi k \sqrt{\|H\|}}{a}$. When utilizing the block-encoding of $\sqrt{\tilde{H}}$, the target polynomial is exactly $\cos(\tau x)$. Alternatively, when utilizing $2\tilde{H}-I$, the target polynomial is $\cos\left(\tau \sqrt{\frac{x+1}{2}}\right)$. According to standard QSVT theorems and the Jacobi-Anger expansion bounds~\cite{10.1145/3313276.3316366}, approximating a trigonometric function oscillating with frequency $\tau$ to a spectral norm error of $\epsilon$ requires a polynomial degree $d$ scaling linearly with the argument $\tau$ and logarithmically with the inverse precision. Consequently, the required polynomial degree, which directly corresponds to the query complexity, evaluates to $d = \mathcal{O}(\tau + \log(1/\epsilon))$, yielding the stated bound.
	\end{proof}
	
	Consequently, the overall simulation accuracy is governed by the trade-off between the aliasing error $\sum_{n \neq 0} \hat{f}(\sqrt{H} + naI)$ and the truncation error. A larger scaling factor $a$ mitigates aliasing by separating spectral copies but necessitates a proportionally larger truncation ratio $K/a$ to resolve high-frequency oscillations. By the Paley-Wiener theorem~\cite{paley1934fourier}, the smoothness of $\hat{f}$ guarantees the rapid decay of $f(k/a)$, thereby efficiently suppressing the truncation error and allowing the algorithm to maintain optimal query complexity.
	
	\subsection{The Discrete Regime: Contour-PSF for Matrix Functions}
	
	The spectral aliasing mechanism described above is not unique to the Fourier domain. To rigorously simulate general matrix functions $f(A)$ where $A$ may not be Hermitian, we extend the PSF framework to finite groups via a novel discrete contour integration scheme. 
	
	Classically, the matrix function is defined via the Cauchy integral formula 
	\begin{equation}
		f(A) = \frac{1}{2\pi i} \oint_{\Gamma} f(z)(zI - A)^{-1} \, dz.
	\end{equation}
	We modify this paradigm by discretizing the contour $C_{R_1}$ (a circle enclosing $\sigma(A)$) using the roots of unity. We analyze the weighted contour integral over an auxiliary outer boundary $C_{R_2}$ ($R_2 > R_1$):
	\begin{equation}
		\mathcal{I} = \frac{1}{2\pi i} \oint_{C_{R_2}} g(z) f(z)(zI - A)^{-1} \, dz,
	\end{equation}
	utilizing a filter kernel $g(z) = \frac{z^{m}}{z^{m}-R_{1}^m}$ designed for the cyclic group $\mathbb{Z}_m$. The poles of $g(z)$ on $C_{R_1}$ form the sampling lattice. By applying the residue theorem, we derive the finite Poisson summation identity:
	\begin{equation}
		\label{eq:finite_poisson_identity}
		\begin{split}
			&\underbrace{\frac{1}{m}\sum_{k=1}^{m} w_{k}f(w_{k})(w_{k}I-A)^{-1}}_{\text{Discrete Sampling (QSVT+LCU)}} 
			= 
			\underbrace{f(A)}_{\text{Target}} -\underbrace{f(A)g(A)}_{\text{Geometric Aliasing}} 
			\\
			&+ 
			\underbrace{\frac{1}{2\pi i} \oint_{C_{R_2}}\frac{R_{1}^{m}}{z^{m}-R_{1}^{m}} f(z)(zI - A)^{-1} \, dz}_{\text{Analytic Truncation}},
		\end{split}
	\end{equation}
	where $w_k = R_1 e^{i \frac{2\pi k}{m}}$, and Eq.~\eqref{eq:finite_poisson_identity} is visualized in Fig.~\ref{fig:contour_schematic}.
	
	\begin{figure}[htbp]
		\centering
		\includegraphics[width=0.45\textwidth]{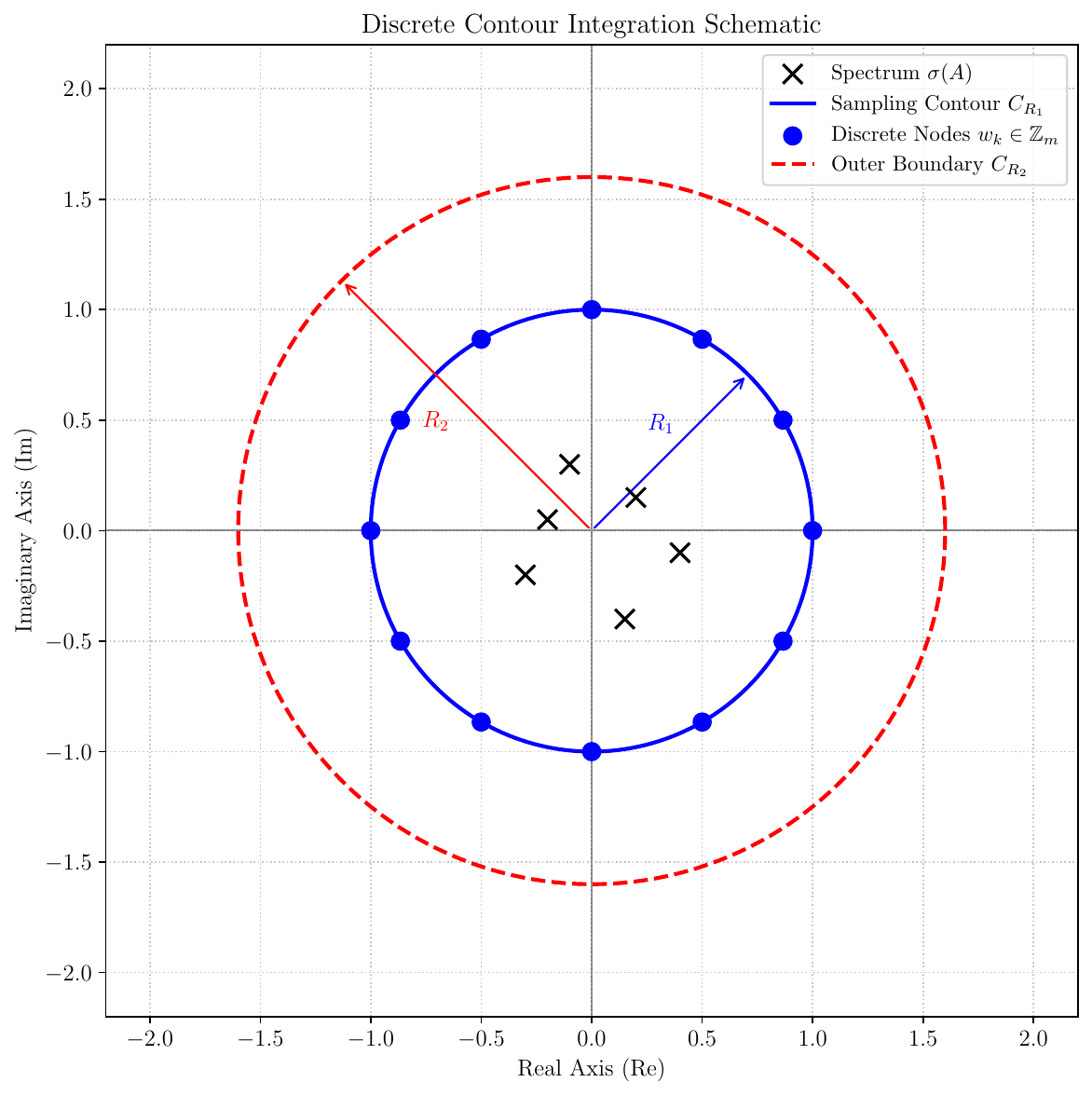}
		\caption{Schematic of the discrete contour integration framework in the complex plane. The black crosses represent the spectrum $\sigma(A)$ enclosed by the sampling contour $C_{R_1}$ (blue). Discrete nodes $w_k$ (blue dots) represent the finite Poisson summation points. The outer boundary $C_{R_2}$ (red dashed line) defines the domain for analytic truncation error estimation as formulated in Eq.~\eqref{eq:finite_poisson_identity}.}
		\label{fig:contour_schematic}
	\end{figure}
	
	This result mirrors Eq.~\eqref{eq:poisson_operator} in the complex domain, decomposing the approximation error into two analogous components. The first is the geometric aliasing error, represented by the term $f(A)g(A)$ arising from the kernel residue at $A$. In the group-theoretic perspective, this constitutes the spectral folding induced by sampling on the subgroup $\mathbb{Z}_m$, corresponding to the summation over the dual lattice (matrix powers $A^{lm}$). The second component is the analytic truncation error, defined by the integral over the outer boundary $C_{R_2}$, which captures the residual error from the finite domain and is analogous to the time-domain tail truncation in the continuous case.
	
	In summary, whether in the continuous Fourier domain or the discrete complex plane, the unified Poisson summation principle allows us to rigorously bound simulation errors by balancing sampling density (suppressing aliasing) against domain size (suppressing truncation).

	\section{Path A: Fourier-PSF for Non-Hermite Dynamics}
	\label{sec:path_a_fourier}

	\subsection{Transmutation as Time-Domain Filtering}
	
	In the specific context of simulating non-Hermitian dissipative dynamics, our subsequent error analysis is strictly grounded in the operator identity derived in Eq.~\eqref{eq:poisson_operator}. This formulation explicitly exposes distinct error sources on both the left-hand side (truncation error) and the right-hand side (aliasing error), implying that the rigorous containment of the total simulation error necessitates a balanced optimization of both the scaling parameter $a$ and the cutoff $K$.
	
	The computational complexity of our framework is governed largely by the \textit{truncation ratio} $K/a$, rather than the scaling parameter $a$ alone. In the context of the LCU, the algorithm involves applying a superposition of time-evolution operators of the form $\cos (
	\frac{2\pi k}{a}\sqrt{H})$. The maximum simulation time required in the quantum circuit is therefore proportional to $\tau_{\max} \approx 2\pi K/a$. Since the query complexity of Hamiltonian simulation scales linearly with the evolution time, specifically as $\mathcal{O}(\sqrt{\|H\|}\tau_{\max})$, the ratio $K/a$ becomes the dominant factor determining the algorithmic cost. While the cutoff number $K$ also influences the complexity of the coefficient state preparation, its impact on the dominant query complexity is mediated entirely through this ratio.
	
	Our analysis reveals that this crucial ratio $K/a$ is intrinsically linked to the regularity of the spectral function $\hat{f}(\xi)$, particularly its differentiability at the origin. When the target spectral function is an entire function, as is the case for integers $\alpha$, the Paley-Wiener theorems guarantee that its inverse Fourier transform $f(x)$ decays exponentially in the time domain~\cite{paley1934fourier}. This rapid decay allows us to truncate the summation at a relatively small $K$, resulting in a minimal ratio $K/a$ and ensuring optimal query complexity. This represents the ideal regime where the analytical smoothness of the spectral function directly translates to the sparsity of the LCU decomposition.
	
	In this study, we restrict our primary focus to cases where $\alpha$ is an even integer. For instances where $\alpha\geq 0.5$ is non-integers, we employ a $y$-axis symmetry construction to design the spectral function $\hat{f}(\xi)$. We acknowledge that this approach may introduce a cusp or limit the order of differentiability at the origin.
	
	\subsection{Analysis of Complexity for $\alpha \geq 0.5$}\label{sec:ana}
	
	When the decay order is chosen as an positive integer (i.e., $\alpha \in \mathbb{Z}^+$), the spectral function $\hat{f}(\xi) = e^{-T\xi^{2\alpha}}$ becomes an entire function in the complex plane. This analyticity implies that its inverse Fourier transform $f(x)$ possesses a exponential decay in the time domain, allowing for highly efficient truncation. We summarize the convergence properties and the resulting query complexity in the following theorem.
	
	\begin{theorem}[Complexity for $\alpha \in \mathbb{Z}^+$] \label{thm:even_alpha_complexity}
		Let $\alpha \in \mathbb{Z}^+$ be an integer, assume we have access to the block encoding of $2\tilde{H}-I$ or $\sqrt{\tilde{H}}$, $H^\alpha=A$ and $H\succeq 0 $. The non-unitary dynamics of $u_T=e^{-T H^{\alpha}} u_0= e^{-T A} u_0$ can be solved by a quantum algorithm with precision $\epsilon$ to get normalized final state $\ket{u_T}$ with $\Omega(1)$ success probability. The query complexity to the matrix oracle for $2\tilde{H}-I$ is
		\begin{equation}\label{eq:analyticquery_int}
			\begin{split}
				\mathcal{O}\Bigg( u_r \log(1+\alpha) \bigg[ \|A\|^{\frac{1}{2\alpha}} T^{\frac{1}{2\alpha}} \left( \log\frac{u_r}{\epsilon} \right)^{1 - \frac{1}{2\alpha}}  + \log\frac{u_r}{\epsilon} \bigg] \Bigg),
			\end{split}
		\end{equation}
		the query complexity to the oracle for initial state preparation is $\Ocal(u_r\log(1+\alpha))$ and the number of LCU coefficients scales as:
		\begin{equation}
			\mathcal{O}\left( \|A\|^{\frac{1}{2\alpha}} T^{\frac{1}{2\alpha}} \left( \log\frac{u_r}{\epsilon} \right)^{1 - \frac{1}{2\alpha}} + \log\frac{u_r}{\epsilon} \right),
		\end{equation}  here $u_r := \|u_0\|/\|u_T\|$. We show the proof in Appendix~\ref{Proof1}.
	\end{theorem}

	In this section, we have established a comprehensive analysis for the regime where the decay order $\alpha$ is a positive integer. By leveraging the PSF, our approach exploits the inherent analyticity and Schwartz-class properties of the spectral function $\hat{f}(\xi)$. This formulation elegantly circumvents the complexities typically associated with the explicit discretization of continuous integrals. Notably, our framework provides a unified treatment that encompasses the specific results derived under the previous work of the case $\alpha=1$~\cite{jin2026transmutationbasedquantumsimulation, kharazi2026sublineartimequantumalgorithmhighdimensional,low2017hamiltoniansimulationuniformspectral}, consistently achieving exponential convergence across the entire integer spectrum.

	In fact, we can directly establish a corollary for the specific case where $\alpha$ is an even integer. In this scenario, we can drop the input assumption regarding $2\tilde{H}-I$ and rely solely on the standard input assumption of $\tilde{H}$, yielding the following Corollary~\ref{cor:even_alpha}:
	
	\begin{corollary}\label{cor:even_alpha}
		For any even positive integer $\alpha \in 2\mathbb{Z}^+$, assuming access solely to the standard block-encoding of $\tilde{H}$ without assumption of $H\succeq 0$, the quantum query complexity to the oracle of $\tilde{H}$ established in Theorem~\ref{thm:even_alpha_complexity} becomes:
		\begin{equation}\label{eq:analyticquery_even}
			\begin{split}
				\mathcal{O}\Bigg( u_r \log(1+\alpha) \bigg[ \|A\|^{\frac{1}{\alpha}} T^{\frac{1}{\alpha}} \left( \log\frac{u_r}{\epsilon} \right)^{1 - \frac{1}{\alpha}}  + \log\frac{u_r}{\epsilon} \bigg] \Bigg).
			\end{split}
		\end{equation}
		Furthermore, the number of coefficients in the LCU scales as:
		\begin{equation}
			\mathcal{O}\left( \|A\|^{\frac{1}{\alpha}} T^{\frac{1}{\alpha}} \left( \log\frac{u_r}{\epsilon} \right)^{1 - \frac{1}{\alpha}} + \log\frac{u_r}{\epsilon} \right),
		\end{equation}
		while the query complexity for the initial state preparation remains identical to that in Theorem~\ref{thm:even_alpha_complexity}.
	\end{corollary}

	We now address the general scenario where the decay order $\alpha$ is a non-integer positive number. This regime encompasses both the heavy-tailed stable distributions (where $0.5\leq\alpha < 2$) and the lighter-tailed but non-analytic dynamics (where $\alpha > 2$). Under the rigorous constraint of preserving the definition of the spectral function $\hat{f}(\xi) = e^{-t|\xi|^{\alpha}}$ for the positive half-axis, the function inevitably encounters a fundamental regularity barrier at the origin. Unlike the integer cases where the function is either analytic or can be regularized to high-order differentiability, a non-integer exponent $\alpha$ introduces a branch point singularity at $\xi=0$. Consequently, the smoothness at the origin is strictly limited by the fractional power, dictating a heavy algebraic tail in the time domain. We provide the corresponding theorem for the case of non-integer functions.
	
	\begin{theorem}[Complexity for $\alpha \notin \mathbb{Z}$ and $\alpha \geq 0.5$] \label{thm:fractional_complexity}
		Let $\alpha \in \mathbb{R}^+ \setminus \mathbb{Z}$ and  be a non-integer. The non-unitary dynamics can be solved by a quantum algorithm with precision $\epsilon$ with $\Omega(1)$ success probability from Eq.~\eqref{eq:origin_problem}, assume we have access to the block encoding of $2\tilde{H}-I$ or $\sqrt{\tilde{H}}$ and $H\succeq 0$. Due to the fractional singularity, the query complexity to the matrix oracle for $H$ scales polynomially with the inverse precision:
		\begin{equation} 
			\mathcal{O}\left( u_r \log(\alpha+e) \alpha \cdot  \|A\|^{\frac{1}{2\alpha}} T^{\frac{1}{2\alpha}} \left( \frac{u_r}{\epsilon} \right)^{\frac{1}{2\alpha}} \right),
		\end{equation}
		where the dominant scaling is $\mathcal{O}(\epsilon^{-\frac{1}{2\alpha}})$. The number of LCU coefficients scales as:
		\begin{equation} 
			\mathcal{O}\left(\alpha \left(\frac{ u_r}{\epsilon}\right)^{\frac{1}{2\alpha}} \left(T^{\frac{1}{2\alpha}}\|A\|^{\frac{1}{2\alpha}}+(\log\frac{u_r}{\epsilon})^{\frac{1}{2\alpha}}\right)\right).
		\end{equation}
		And query complexity to $|u_0\rangle$ is about $\Ocal(u_r\log(e+\alpha))$, here $ u_r := \|u_0\|/\|u_T\|$. We show the proof in Appendix~\ref{Proof3}.
	\end{theorem}
	
	\begin{figure}[htbp]
		\centering
		\includegraphics[width=0.45\textwidth]{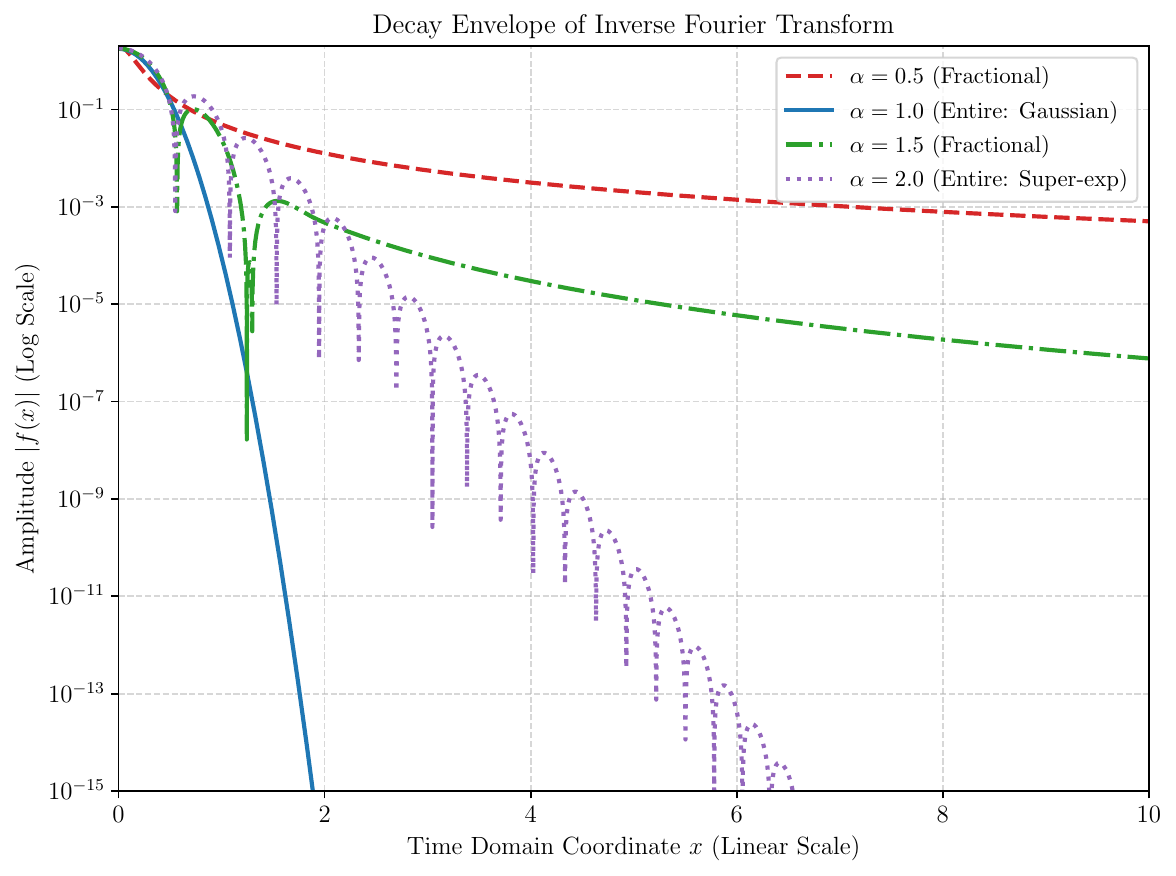}
		\caption{Decay envelopes of the time-domain integral $f(x)$ for different dissipation orders $\alpha$ on a linear x-axis. For integer values ($\alpha=1.0$ and $2.0$), the function exhibits rapid exponential or super-exponential decay, whereas fractional orders ($\alpha=0.5$ and $1.5$) result in algebraic heavy tails $|f(x)| \sim |x|^{-(2\alpha+1)}$, validating the complexity transitions established in Theorem~\ref{thm:even_alpha_complexity} and Theorem~\ref{thm:fractional_complexity}.}
		\label{fig:fourier_decay}
	\end{figure}
	
	The simulation of fractional decay orders $\alpha \notin \mathbb{Z}$ reveals an intrinsic computational barrier imposed by the non-analytic nature of the operator. Crucially, if we restrict our access solely to the standard block-encoding oracle of $\tilde{H}$ (foregoing the domain-shifted $2\tilde{H}-I$ oracle), this theoretical framework expands: the computational barrier now encompasses odd integer values of $\alpha$ in addition to fractional orders. Under this restricted input assumption, all scaling exponents of $2\alpha$ established in Theorem~\ref{thm:fractional_complexity} are intrinsically modified to $\alpha$. Specifically, the relevant power in the spectral function becomes $|\xi|^{\alpha}$, which introduces a branch point singularity (or lack of analyticity) at the origin that dictates a heavy algebraic tail $|x|^{-(\alpha+1)}$ in the time domain as shown in Fig.~\ref{fig:fourier_decay}. 
	
	This slow decay is fundamental. It cannot be circumvented by local smoothing techniques without altering the nonlocal physics of the fractional dynamics. Consequently, the algorithmic complexity transitions from the efficient quasi-polylogarithmic scaling observed in analytic (even integer) regimes to a polynomial scaling $\mathcal{O}(\epsilon^{-\frac{1}{\alpha}})$, highlighting a distinct trade-off between the physical anomalousness of the diffusion and the computational resources required to simulate it.
	
	\section{Path B: Contour-PSF for Holomorphic Matrix Functions}
	\label{sec:path_b_contour}
	
	In this section, we rigorously derive the discretization of the contour integral using the framework of the finite PSF. Unlike standard quadrature error analysis which often yields loose polynomial bounds, our group-theoretic approach reveals the exponential convergence of the approximation. Furthermore, we analyze the complexity of the LCU implementation and propose an optimization strategy for the contour radius.
	
	\subsection{Algorithm Framework}
	
	Leveraging Poisson summation on finite groups, we derive a novel identity for general quantum eigenvalue transformations, presented in Eq.~\eqref{eq:finite_poisson_identity}. This formulation circumvents the conventional contour discretization analysis required in previous approaches, enabling a direct evaluation of the discrete sampling term and the associated error terms.
	
	Notably, the discrete sampling term aligns perfectly with those found in recent literature~\cite{jiang2026contourintegralbasedquantumeigenvalue}, implying that the underlying circuit implementation remains unchanged, our primary contribution lies in providing a more transparent and rigorous analytical framework. We corroborate the established conclusion that the number of sampling points of $m$, which not impact the asymptotic query complexity of the input matrix. However, $m$ significantly influences the overall circuit complexity, specifically regarding the state preparation for the LCU and the cumulative overhead incurred by amplitude amplification.
	
	Eq.~\eqref{eq:finite_poisson_identity} explicitly characterizes the approximation error through two components: the geometric aliasing term, $-f(A)g(A)$, and the analytic truncation integral over $C_{R_2}$. Both terms vanish asymptotically as $m \to \infty$. Furthermore, we demonstrate that to achieve a target precision $\epsilon'$, the required number of sampling points scales as $m \sim \tOcal(\log(\frac{1}{\epsilon'}))$. This represents an exponential improvement in the scaling estimate of $m$ compared to prior studies.
	
	\subsection{Complexity Analysis and Improvement}\label{Sec:facomplex}
	Regarding Eq.~\eqref{eq:finite_poisson_identity}, we approximate the target operator $f(A)$ by implementing the discrete sum $\frac{1}{m}\sum_{k=1}^{m} w_{k}f(w_{k})(w_{k}I-A)^{-1}$. This is achieved by employing the QSVT to approximate the matrix inversion process. Building upon established QSVT theorems and prior research, we utilize the block-encoding matrix $\frac{w_{k}I-A^{\dagger}}{\alpha_A+R_{1}}$ to perform the QSVT, thereby approximating the scaled inverse term corresponding to $(w_{k}I-A)^{-1}$. To achieve an approximation error of $\epsilon'$, the required query complexity is given by $\mathcal{O}\left(\gamma (\alpha_A+R_{1})\log\frac{1}{\epsilon'}\right)$, where $\gamma$ denotes the upper bound on the spectral norm of the resolvent $(zI-A)^{-1}$ on $C_{R_1}$, and $\alpha_A \geq \|A\|$.
	
	Furthermore, considering the operation is applied to a quantum state $|\psi\rangle$, amplitude amplification is required to ensure a final success probability of $\Omega(1)$. The corresponding amplification factor is given by:
	\begin{equation}
		\mathcal{O}\left(\frac{4\gamma}{3}\frac{\frac{1}{m}\sum_{k=1}^{m} |w_{k}f(w_{k})|}{\|f(A)|\psi\rangle\|}\right)
		= \mathcal{O}\left(\gamma\frac{R_{1}B_{1}}{\|f(A)|\psi\rangle\|}\right),
	\end{equation}
	where $B_1$ represents the bound of the function $|f(z)|$ on the $C_{R_1}$ integration contour. 
	
	Theoretically, we can bounded $R_1$ by $\alpha_A$. Under this premise, the query complexity for matrix $A$ is initially given by:
	\begin{equation}
		\mathcal{O}\left(\frac{\gamma^2\alpha_A^2B_{1}}{\|f(A)|\psi\rangle\|}\log\frac{1}{\epsilon'}\right).
	\end{equation}
	Considering that amplitude amplification scales the error, and aiming to bound the final error contribution from this part by $\frac{\epsilon}{2}$, we impose the condition:
	\begin{equation}
		\epsilon'\gamma\frac{\alpha_A B_{1}}{\|f(A)|\psi\rangle\|}=\frac{\epsilon}{2}.
	\end{equation}
	The remaining error budget of $\frac{\epsilon}{2}$ is allocated to control the discretization size $m$. Consequently, we obtain the final query complexity for matrix $A$ as 
	\begin{equation}
		\mathcal{O}\left(\frac{\gamma^2\alpha_A^2B_{1}}{\|f(A)|\psi\rangle\|}\log\frac{\gamma \alpha_A B_{1}}{\|f(A)|\psi\rangle\|\epsilon}\right), 
	\end{equation}
	and the complexity for preparing the quantum state $|\psi\rangle$ as
	\begin{equation}
		\mathcal{O}\left(\frac{\gamma \alpha_A B_{1}}{\|f(A)|\psi\rangle\|}\right).
	\end{equation}
	
	The selection of the sampling parameter $m$ constitutes the primary enhancement in the transformation of our formula Eq.~(\ref{eq:finite_poisson_identity}). Notably, the determination of $m$ is independent of the algorithm's query complexity, allowing for a virtually decoupled analysis. Given the selection of the normalization coefficient $\|f(A)|\psi\rangle\|$ and the allocated error budget of $\frac{\epsilon}{2}$, we ultimately require the following condition to hold:
	\begin{equation}
		\begin{split}
			&\left\| \frac{1}{2\pi i} \oint_{C_{R_2}}\frac{R_{1}^{m}}{z^{m}-R_{1}^{m}}\frac{f(z)}{zI - A}  \, dz|\psi\rangle - f(A)g(A)|\psi\rangle \right\|\\& \leq \frac{\|f(A)|\psi\rangle\| \epsilon}{2}.
		\end{split}
	\end{equation}
	This requirement can be analyzed through the following two aspects. First we need:
	\begin{equation}
		\norm{	f(A)g(A)|\psi\rangle}=\norm{ \frac{A^{m}}{A^{m}-R_{1}^{m}}f(A)|\psi\rangle}\leq  \frac{\|f(A)|\psi\rangle\| \epsilon}{4},
	\end{equation}
	then we get $m\in\Ocal\left(\log\frac{1}{\epsilon} \right)$. 
	
	Next, we estimate the upper bound of the integral. We introduce the following parameters defined on the contour $C_{R_2}$: let $B_2 = \max_{z \in C_{R_2}} |f(z)|$ denote the maximum magnitude of the function, and let $\gamma_2 = \max_{z \in C_{R_2}} \|(zI - A)^{-1}\|$ denote the upper bound of the resolvent norm, and define the ratio $\mu = R_1/R_2 < 1$.
	
	Given that the spectrum of $A$ is strictly contained within the radius $R_1$, the distance between any point $z$ on the contour $C_{R_2}$ and the eigenvalues is lower bounded by $R_2 - R_1$. Consequently, for a diagonalizable matrix $A$, the resolvent norm $\gamma_2$ can be estimated involving the condition number $\kappa_S$ of the eigenvector matrix:
	\begin{equation}
		\gamma_2 = \max_{z \in C_{R_2}} \|(zI - A)^{-1}\| \leq \frac{\kappa_S}{R_2 - R_1}.
	\end{equation}
	
	By applying the standard estimation lemma (ML inequality) and substituting the bound for $\gamma_2$, the norm of the integral term can be bounded as follows:
	\begin{equation}
		\begin{split}
			&\left\| \frac{1}{2\pi i} \oint_{C_{R_2}}\frac{R_{1}^{m}}{z^{m}-R_{1}^{m}} \frac{f(z)}{zI - A} \, dz|\psi\rangle \right\|\\
			\leq &\frac{1}{2\pi} (2\pi R_2) \frac{\mu^m}{1-\mu^m} B_2 \gamma_2 \|\psi\|\\ 
			\leq &\frac{R_2 B_2 \kappa_S \|\psi\| \mu^m}{(1-\mu^m)(R_2 - R_1)}.
		\end{split}
	\end{equation}
	
	To satisfy the target error bound of $\frac{\|f(A)|\psi\rangle\| \epsilon}{4}$, the sampling parameter $m$ must satisfy:
	\begin{equation}
		\frac{\mu^m}{1-\mu^m} \leq \frac{\|f(A)|\psi\rangle\| \epsilon (R_2 - R_1)}{4 R_2 B_2 \kappa_S \|\psi\|}.
	\end{equation}
	Solving for $m$, and combine the term of $\norm{	f(A)g(A)|\psi\rangle}$ we obtain the asymptotic scaling:
	\begin{equation}\label{eq:mbound}
		m = \mathcal{O}\left( \frac{\log(\frac{\kappa_S B_2}{\epsilon})}{\log(R_2/R_1)} \right).
	\end{equation}
	This indicates that $m$ scales logarithmically with the inverse precision $1/\epsilon$ and the condition number, and inversely with the logarithm of the gap ratio $R_2/R_1$.
	\section{Applications}

	\subsection{High-Order Dissipation: Biharmonic Diffusion}
	
	We consider the homogeneous biharmonic diffusion equation, a fundamental model in continuum mechanics for thin film evolution and the Cahn-Hilliard dynamics of phase separation~\cite{10.10631.1722742}:
	\begin{equation}
		\partial_t u(x,t) = -\Delta^2 u(x,t), \quad u(x,0) = u_0(x).
	\end{equation}
	In our unified framework, this high-order dissipation corresponds to a decay order of $\alpha=4$.
	
	Conventional transmutation-based methods, such as those proposed by Jin et al.~\cite{jin2026transmutationbasedquantumsimulation}, typically rely on explicitly decomposing the generator into a symmetric product form to simulate high-order operators. To implement our underlying spatial discretization, we naturally adopt the elegant finite-difference matrix construction introduced in their work. Specifically, we utilize their auxiliary first-order difference operator $L'$ defined as~\cite{jin2026transmutationbasedquantumsimulation}:
	\begin{equation}
		L^{\prime} :=\frac{1}{h} \begin{bmatrix}
			-1 & & &   \\
			1 & -1 & &  \\
			& \ddots & \ddots & \\
			& & 1 & -1  \\
			&  & & 1  
		\end{bmatrix}.
	\end{equation}
	While conventional approaches might require assembling increasingly complex higher-order discrete operators to simulate phenomena like biharmonic diffusion, our Fourier-PSF framework circumvents this overhead. We achieve a structurally simpler implementation and superior algorithmic complexity by lifting the high-order dynamics entirely into the classical spectral weight function.
	
	For a $d$-dimensional system, the full discrete gradient operator $L$ is constructed by stacking the directional derivatives:
	\begin{equation}
		L := \left[ (L^{(1)})^\dagger, (L^{(2)})^\dagger, \cdots, (L^{(d)})^\dagger \right]^\dagger,
	\end{equation}
	where $L^{(k)}$ acts as $L'$ along the $k$-th dimension and as the identity elsewhere~\cite{jin2026transmutationbasedquantumsimulation}. We then define the Hermitian root operator (a Dirac-like operator) $H$ as:
	\begin{equation}
		H := \begin{bmatrix}
			O & -iL^{\dagger}   \\
			iL & O  
		\end{bmatrix}.
	\end{equation}
	Instead of constructing a new, dense high-order discrete operator for $\Delta^2$, we directly utilize $H$. It is straightforward to verify that the primary block of $H^4$ naturally recovers the discrete biharmonic operator $(L^\dagger L)^2 \approx \Delta^2$, incorporating all mixed derivative terms required for isotropic diffusion. Thus, the target non-unitary evolution is directly given by $e^{-tA} \approx e^{-tH^4}$.
	
	Since $\alpha=2$ is an integer, the corresponding spectral function $\hat{f}(\xi) = e^{-t \xi^{2\alpha}}$ is an entire function. As established in Sec.~\ref{sec:fpsf} and Sec.~\ref{sec:ana}, the algorithm resides in the analytic regime governed by Corollary~\ref{cor:even_alpha}, achieving exponential convergence with respect to the truncation radius. Consider $\Omega=(0,1)^{d}$ and the staggered finite-difference discretization described above. With the spectral norm scaling as $\|H\| = \Theta(\sqrt{d}/h)$, the quantum query complexity to the block-encoding oracle to prepare an $\epsilon$-approximation of the normalized state scales as:
	\begin{equation}
		\mathcal{O}\left( u_r \left( \frac{\sqrt{d}}{h} T^{\frac{1}{4}} \left(\log \frac{u_r}{\epsilon}\right)^{\frac{3}{4}} + \log \frac{u_r}{\epsilon} \right) \right)
	\end{equation}
	along with $\mathcal{O}(u_r)$ queries to the initial state preparation oracle. By lifting the dynamics into the spectral weights, our method simulates fourth-order dissipation while keeping the hardware implementation complexity equivalent to that of a simple first-order wave equation solver.
	
	Furthermore, it is worth noting that this identical framework naturally recovers the simulation of standard second-order dissipation (i.e., the homogeneous heat equation, $\partial_t u = \Delta u$). Setting $\alpha=1$ yields the target evolution $e^{-tH^2}$ and the spectral function $\hat{f}(\xi) = e^{-t \xi^2}$. Following the same analytical procedure, the oracle query complexity for the heat equation strictly bounds to $\mathcal{O}\left( u_r \left( \frac{\sqrt{d}}{h} T^{\frac{1}{2}} \left(\log \frac{u_r}{\epsilon}\right)^{\frac{1}{2}} + \log \frac{u_r}{\epsilon} \right) \right)$, which matches the optimal asymptotic scaling established in previous state-of-the-art works~\cite{jin2026transmutationbasedquantumsimulation}.

	\subsection{Super-diffusive Transport: L\'evy Flights}
	
	We focus on a specific instance of anomalous diffusion corresponding to the Holtsmark distribution class~\cite{METZLER20001}. The governing equation models super-diffusive transport phenomena, such as turbulent flows and financial market fluctuations~\cite{LISCHKE2020109009}. The dynamics are described by the fractional diffusion equation:
	\begin{equation}
		\partial_t u(x,t) = -(-\Delta)^{0.75} u(x,t), \quad u(x,0) = u_0(x).
	\end{equation}
	In our unified framework, the target non-unitary evolution is generated by $e^{-t(-\Delta)^{0.75}}$. 
	
	Applying conventional transmutation methods to this specific problem is exceptionally difficult. A standard factorization would require identifying a local operator for the fractional power $-(-\Delta)^{0.75}$. Unlike the even integer-order Laplacian, this operator is intrinsically non-local; its spatial discretization results in a dense matrix. Implementing a quantum block-encoding for such a dense Hamiltonian incurs a prohibitive gate complexity, effectively nullifying the quantum advantage.

	In contrast, our sparse block-encoding framework handles this scenario without altering the underlying quantum hardware architecture or requiring dense matrix embeddings. By directly setting the root operator to the standard discrete Laplacian, $H = -\Delta$, which is a highly sparse difference operator, we can construct the target evolution $e^{-t H^{0.75}}$ using the $2\tilde{H}-I$ encoding scheme. Fast and highly efficient quantum circuits for block-encoding such structured, sparse matrices are well-established~~\cite{10.1145/3313276.3316366, Sunderhauf2024blockencoding, Camps_2022}. 
	
	For a $d$-dimensional spatial grid with mesh size $h$, the spectral norm is $\|H\| \approx 4d/h^2$. The shifted matrix $A = 2H/\|H\| - I$ consistently maintains a perfectly vanishing diagonal and a row 1-norm strictly equal to 1. This elegant structural guarantee is not a mere coincidence; it arises because $A$ is isomorphic to the negative of the transition matrix of a simple random walk on the $d$-dimensional grid~\cite{1366222}. Consequently, this strictly avoids any additional state preparation normalization overhead (i.e., the sub-normalization factor is exactly 1).
	
	Crucially, because this encoding inherently evaluates the polynomial $\cos\left(\tau\sqrt{\frac{x+1}{2}}\right)$, it enables the direct, hardware-efficient implementation of $\cos\left(\frac{2\pi k}{a}\sqrt{H}\right)$ without querying a dense fractional operator. To match the target evolution, the fractional nature of the dynamics is entirely absorbed into the classical coefficient calculation by defining the effective spectral function:
	\begin{equation}
		c_k = \frac{1}{a} f\left(\frac{k}{a}\right) \quad \text{with} \quad \hat{f}(\xi) = e^{-t|\xi|^{2 \times 0.75}} = e^{-t|\xi|^{1.5}}.
	\end{equation}
	
	Because the effective non-integer exponent $2\alpha=1.5$ places this problem in the fractional regime, the spectral function contains a branch point singularity at $\xi=0$. As established by Tauberian theorems, this dictates a heavy-tailed algebraic decay in the time domain:
	\begin{equation}
		|f(x)| \sim \frac{1}{|x|^{2\alpha+1}} = \frac{1}{|x|^{2.5}}.
	\end{equation}
	
	While this slow decay necessitates a larger truncation radius scaling as $K/a \sim \epsilon^{-\frac{1}{2\alpha}} = \epsilon^{-\frac{2}{3}}$ compared to the exponential convergence of integer-order cases, it provides a unique and viable pathway to simulate Holtsmark-type dynamics using only standard local connectivity on the quantum processor. 
	
	Applying Theorem~\ref{thm:fractional_complexity}, we obtain a quantum algorithm to prepare an $\epsilon$-approximation of the normalized solution state $\ket{u_{h}(T)}$ with $\Omega(1)$ success probability. Because our formulation directly evaluates $\sqrt{H}$, the maximum phase in the QSVT relies strictly on the square root of the spectral norm, $\sqrt{\|H\|} = \mathcal{O}(\sqrt{d}/h)$, rather than the full norm $\|H\| = \mathcal{O}(d/h^2)$. Substituting $\alpha=0.75$, the query complexity to the block-encoding oracle of the Laplacian $H$ scales polynomially with the inverse precision:
	\begin{equation}\begin{split} 
			\mathcal{O}\left( \frac{\sqrt{d}}{h} u_r^{\frac{5}{3}} \left(\frac{T}{\epsilon}\right)^{\frac{2}{3}} \right),
	\end{split}\end{equation}
	where $d$ is the spatial dimension and $h$ is the mesh size. The query complexity for the initial state preparation oracle is $\mathcal{O}(u_r)$. 
	Furthermore, the number of LCU coefficients required scales as:
	\begin{equation}
		\mathcal{O}\left( \left(\frac{u_r}{\epsilon}\right)^{\frac{2}{3}} \left( \frac{\sqrt{d}}{h} T^{\frac{2}{3}} + \left(\log \frac{u_r}{\epsilon}\right)^{\frac{2}{3}} \right) \right).
	\end{equation}

	\subsection{Matrix Polynomials}
	\label{sec:matrix_poly}
	
	We consider the implementation of a generic matrix polynomial $f(A) = \sum_{j=0}^{d} c_j A^j$. 
	To accommodate the spectrum of $A$, we utilize the block-encoding factor $\alpha_A$ ($\alpha_A \ge \|A\|$) to set the inner radius $R_1 \approx \alpha_A$. 
	For the integration contour $\Gamma$, we select a circle of radius $R_2$ strictly larger than $R_1$ ($R_2 > R_1$), but leave the ratio $R_2/R_1$ as a tunable parameter.
	
	In this general configuration, the key geometric parameters governing the algorithm's performance are interrelated. 
	Specifically, the spectral gap, representing the distance between the contour and the singularity, which is $ R_2 - R_1$. 
	This gap directly determines the upper bound of the resolvent norm $\gamma_2$ for a diagonalizable matrix $A=SDS^{-1}$, which is estimated as:
	\begin{equation}
		\gamma_2 = \max_{\abs{z}={R_2}} \|(zI - A)^{-1}\| \le \frac{\kappa_S}{R_2 - R_1},
	\end{equation}
	where $\kappa_S$ is the condition number of the eigenvector matrix. 
	Complementing this, the length of the integration contour is given by $l = 2\pi R_2$.
	
	The maximum magnitude of the function on the contour is defined as $B_2 = \max_{|z|=R_2} |f(z)|$. 
	Unlike previous analyses that bind $B_2$ strictly to the polynomial degree, we treat $B_2$ as a property dependent on the chosen radius $R_2$. 
	This allows for a flexible analysis applicable to polynomials with varying coefficient decay rates.
	
	Regarding the discretization error, the required number of sampling nodes $m$ is decoupled from the query complexity and is determined by the convergence rate of the geometric series $(R_1/R_2)^m$. 
	To satisfy the error tolerance $\epsilon$, $m$ satisfy:
	\begin{equation}
		m\sim \Ocal\left(\frac{\log(B_2/\epsilon)}{\log(R_2/R_1)}\right).
	\end{equation}
	This expression highlights that $m$ scales logarithmically with the function magnitude $B_2$ and inversely with the logarithm of the geometric ratio $R_2/R_1$. 
	While a smaller gap $R_2 - R_1$ increases the required $m$, this only impacts the classical pre-processing and the LCU control depth logarithmically, without increasing the number of queries to $U_A$.
	
	Finally, substituting these general parameters into the total query complexity theorem, we obtain the cost for implementing $f(A)$:
	\begin{equation}
		\tilde{\mathcal{O}}\left( \frac{B_2 \cdot l \cdot \gamma_2^2 \cdot \alpha_A}{\|f(A)|\psi\rangle\|} \right) 
		= \tilde{\mathcal{O}}\left( \frac{B_2 \cdot R_2 \cdot \kappa_S^2 \cdot \alpha_A}{(R_2-R_1)^2 \|f(A)|\psi\rangle\|} \right).
	\end{equation}
	This general bound reveals the intrinsic trade-off in contour selection: increasing the contour radius $R_2$ reduces the resolvent norm term ($\propto (R_2-R_1)^{-2}$), but typically increases the function magnitude $B_2$. 
	Therefore, the optimal $R_2$ is not necessarily fixed at $2R_1$ but should be chosen to minimize the product $\frac{B_2 R_2}{(R_2-R_1)^2}$ based on the specific growth behavior of the target polynomial.

	\section{DISCUSSION AND OUTLOOK}
	
	The unified PSF presented in this work fundamentally reinterprets quantum simulation errors as spectral aliasing in dual domains. By bridging the continuous time and frequency domain (Fourier-PSF) and the discrete complex spectral domain (contour-PSF), we have rigorously resolved the intrinsic smoothness-sparsity tradeoff that dictates algorithmic efficiency. This perspective not only unifies the treatment of fractional singularities and holomorphic functions under a single theoretical umbrella but also reveals deep structural connections between classical harmonic analysis and quantum matrix arithmetic.
	
	A profound implication of our Fourier-PSF approach (Path A) is its ability to achieve lower query complexity for generalized dissipative dynamics than conventional direct QSVT approaches. Standard QSVT predominantly relies on Chebyshev polynomial approximations. These are minimax optimal for oscillatory Hamiltonian simulations but can be suboptimal for nonunitary rapidly decaying functions. The fact that mapping the problem to a Fourier coefficient-weighted LCU yields superior scaling strongly suggests the existence of a fundamentally new polynomial basis. If such a novel basis can be explicitly constructed to be inherently tailored for dissipative envelopes rather than oscillatory dynamics, it could enable polynomial expansions that approximate dissipative Hermitian matrix equations at significantly lower degrees. This would potentially redefine the theoretical lower bounds for quantum nonunitary simulation complexity. Within this path, it is also necessary to clarify our physically motivated constraint of $\alpha \ge 0.5$. As the fractional order $\alpha \to 0$, the inverse scaling in the algorithmic complexity bounds induces a prohibitive computational bottleneck due to the severe mathematical singularity at the origin. Fortunately, such extreme sub-diffusive regimes lack widespread physical correspondence; thus, the $\alpha \ge 0.5$ bound naturally delineates the regime where quantum simulation is both computationally efficient and physically meaningful.
	
	Furthermore, our contour-PSF approach (Path B) currently leverages the cyclic group of roots of unity to perform a finite Poisson summation. This geometrically corresponds to uniform discrete sampling on a circular contour centered at the origin. While this group-theoretic formulation guarantees exponential convergence for spectra bounded within a standard disk, the circular geometry may not be strictly optimal for physical operators possessing highly asymmetric, elongated, or disconnected spectra. This inherent geometric limitation points to a highly promising future algorithmic direction: the integration of complex conformal mappings. By employing analytic tools such as Faber polynomials, it is mathematically viable to conformally map the exterior of an arbitrary simply connected domain (which tightly encloses the specific operator spectrum) onto the exterior of the unit disk. For instance, to accommodate an elongated spectrum, the standard circular kernel $g(z) = z^m / (z^m - R_1^m)$ can be elegantly generalized to an elliptical filter kernel $g_{\mathrm{elp}}(z) = F_m(z) / (F_m(z) - \rho^m)$, where $F_m(z)$ denotes the $m$-th order Faber polynomial. By enforcing $F_m(z) = \rho^m$, the uniformly distributed roots of unity are mapped into non-uniform discrete sampling nodes that naturally cluster at the high-curvature extremities of the physical ellipse (see Fig.~\ref{fig:ellipse_sampling}). This transformation would drastically minimize the resolvent norm bound and extend the exponential convergence guarantees of the contour-PSF to virtually any complex spectral geometry, exponentially suppressing geometric aliasing for narrow spectra.
	
	\begin{figure}[htbp]
		\centering
		\includegraphics[width=0.75\linewidth]{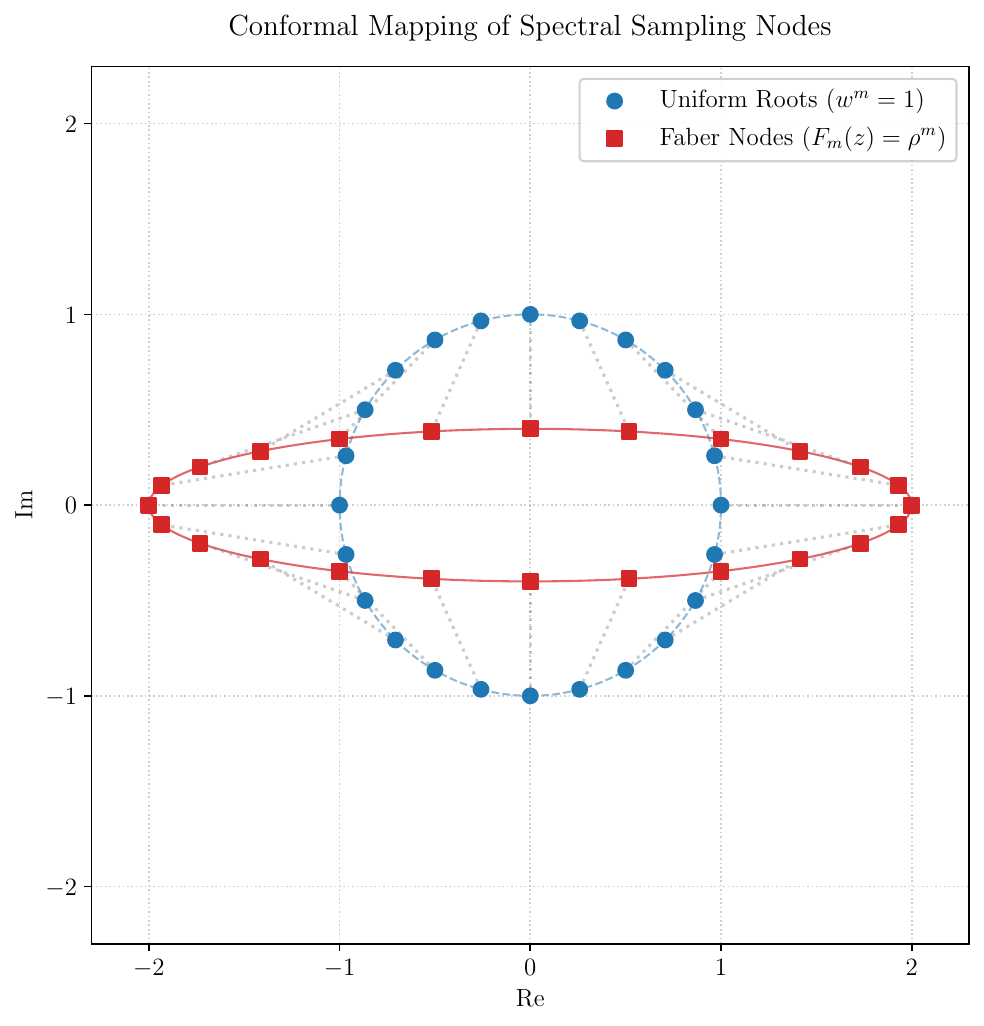}
		\caption{Geometric illustration of the conformal mapping strategy extending the contour-PSF framework. The uniformly distributed roots of unity on the reference unit circle (blue circles, $w^m = 1$) are conformally mapped to non-uniform discrete sampling nodes on the physical elliptical contour (red squares, corresponding to the roots of the Faber polynomial $F_m(z) = \rho^m$). The dotted lines trace the Joukowsky transformation, visually demonstrating how uniform phase intervals in the dual domain inherently cluster the sampling nodes at the high-curvature extremities of the elongated spectrum in the primary domain. This strategic node density dramatically minimizes the resolvent norm bound and exponentially suppresses geometric aliasing for highly asymmetric or elongated spectra, overcoming the limitations of standard circular sampling.}
		\label{fig:ellipse_sampling}
	\end{figure}
	
	Finally, by lifting the dynamical complexity from the structural matrix level into the classical spectral weight functions, our framework demonstrates exceptional hardware friendliness. The ability to bypass the construction of dense or high-order discrete operators while simulating complex superdiffusive and biharmonic phenomena using only standard local quantum primitives effectively decouples the complexity of the physical phenomenon from the hardware connectivity constraints. Future extensions of this spectral duality will naturally explore multidimensional partial differential equations and the application of discrete contour-PSF to multimatrix functions, fully unlocking the potential of Poisson summation in the design of next-generation quantum algorithms.
	
	\begin{acknowledgments}
		This work has been supported by the National Key Research and Development Program of China (Grant No. 2023YFB4502500). We are also grateful to Prof. Dong An for insightful discussions and valuable suggestions. 
	\end{acknowledgments}

	\bibliography{ref}
	
	\newpage 
	\clearpage
	\thispagestyle{empty}
	\onecolumngrid
	\appendix
	
	\section{Estimation of the $L^1$-norm of Function $f$}\label{App:EstimL1}
	
	A bound of $\mathcal{O}(\log \alpha)$ can be derived by analyzing the specific pointwise behavior of $f(x)$ in the limit of large $\alpha$. This improvement arises because $\hat{f}(\xi)$ approaches a rectangular function, and the $L^1$-norm of its inverse transform (Sinc-like kernel) is known to scale logarithmically with the smoothness bandwidth.
	
	\begin{theorem}[Logarithmic Scaling of $L^1$-Norm]\label{thm:L1}
		For the spectral function $\hat{f}(\xi) = e^{-t|\xi|^{2\alpha}}$, as the decay order $\alpha \to \infty$, the $L^1$-norm of the coefficients scales logarithmically:
		\begin{equation}
			\|f\|_{L^1} \approx \mathcal{O}(\log \alpha).
		\end{equation}
	\end{theorem}
	
	\begin{proof}
		We analyze the structure of $f(x)$ by decomposing the frequency domain into a flat region and a transition region. For simplicity, we assume unit time $t=1$ and standard scaling, as the $L^1$-norm is scale-invariant.
		
		As $\alpha \to \infty$, the function $\hat{f}(\xi) = e^{-|\xi|^{2\alpha}}$ converges pointwise to the rectangular function:
		\begin{equation}
			\lim_{\alpha \to \infty} \hat{f}(\xi) = \text{rect}(\xi/2) = \begin{cases} 1, & |\xi| < 1 \\ 0, & |\xi| > 1 \end{cases}.
		\end{equation}
		The inverse Fourier transform of the rectangular function is the normalized Sinc function:
		\begin{equation}
			f_{\infty}(x) = \int_{-1}^{1} e^{2\pi i x \xi} d\xi = \frac{\sin(2\pi x)}{\pi x}.
		\end{equation}
		The integral of the absolute value of the ideal Sinc function diverges logarithmically. However, for finite $\alpha$, $\hat{f}(\xi)$ is not perfectly discontinuous; it possesses a finite transition width.
		
		This finite steepness of $\hat{f}(\xi)$ acts as a high-frequency filter that suppresses the tail of $f(x)$. The steepness of the transition at $|\xi| \approx 1$ can be characterized by the maximum derivative:
		\begin{equation}
			\max |\hat{f}'(\xi)| \approx \left. \frac{d}{d\xi} e^{-\xi^{2\alpha}} \right|_{\xi \approx 1} \approx 2\alpha e^{-1} \sim \mathcal{O}(\alpha).
		\end{equation}
		This implies that the transition width in the frequency domain is $\Delta \xi \sim 1/\alpha$. According to the uncertainty principle, a feature of width $\Delta \xi$ in frequency induces decay in time starting at a critical scale $X_c \sim 1/\Delta \xi \approx \alpha$~\cite{stein1971introduction}.
		
		Specifically, we approximate $f(x)$ in two regimes: the Sinc regime ($|x| \le \alpha$) and the decay regime ($|x| > \alpha$). In the Sinc regime, the behavior is dominated by the sharp spectral cutoff, and the function closely mimics the ideal Sinc function:
		\begin{equation}
			|f(x)| \approx \left| \frac{\sin(2\pi x)}{\pi x} \right|.
		\end{equation} 
		Conversely, in the decay regime ($|x| > \alpha$), the finite transition width $\Delta \xi$ takes effect. The smoothness (differentiability) of $\hat{f}(\xi)$ causes $|f(x)|$ to decay rapidly, ensuring that the tail integral $\int_{|x| > \alpha} |f(x)| dx$ is bounded by a constant $\mathcal{O}(1)$.
		
		We estimate the dominant contribution to the $L^1$-norm by integrating the Sinc envelope up to the effective cutoff $X_c \sim \alpha$. By averaging the highly oscillatory term $|\sin(2\pi x)|$ over its period (which yields a factor of $2/\pi$), we obtain:
		\begin{equation}
			\|f\|_{L^1} \approx \int_{-X_c}^{X_c} |f(x)| dx \approx 2 \int_{1}^{\alpha} \frac{2/\pi}{\pi x} dx + \mathcal{O}(1).
		\end{equation}
		(The $\mathcal{O}(1)$ term accounts for the central peak near $x=0$ and the rapidly decaying tails). Calculating the integral yields:
		\begin{equation}\label{eq:L1est}
			\|f\|_{L^1} \approx \frac{4}{\pi^2} \log(\alpha) + \mathcal{O}(1) \sim \mathcal{O}(\log \alpha).
		\end{equation}
	\end{proof}
	
	This $\mathcal{O}(\log \alpha)$ bound is tighter than the $\mathcal{O}(\alpha^{1/4})$ bound derived from Carlson's inequality~\cite{katznelson2004introduction}. While Carlson's bound is rigorous and sufficient for convergence, the logarithmic scaling reflects the true asymptotic cost of the LCU decomposition for boxcar-like spectral filters.
	
	While the $L^1$-norm exhibits logarithmic growth for large $\alpha$, its behavior is fundamentally different in the low-decay regime where $0 < \alpha \le 1$. In this interval, the spectral function $\hat{f}(\xi) = e^{-t|\xi|^{2\alpha}}$ corresponds to the characteristic function of a symmetric Lévy $\beta$-stable distribution, with stability parameter $\beta = 2\alpha \in (0, 2]$~\cite{grubbs1967introduction}.
	
	A fundamental property of stable distributions with stability parameter $\beta \in (0, 2]$ is that their corresponding probability density functions $f(x)$ are strictly non-negative everywhere:
	\begin{equation}
		f(x) \ge 0, \quad \forall x \in \mathbb{R}.
	\end{equation}
	This positivity property drastically simplifies the $L^1$-norm calculation. Since $|f(x)| = f(x)$, the norm is determined directly by the zero-frequency component of the spectral function via the fundamental property of the Fourier transform:
	\begin{equation}
		\|f\|_{L^1} = \int_{-\infty}^{\infty} |f(x)| \, dx = \int_{-\infty}^{\infty} f(x) \, dx = \hat{f}(0).
	\end{equation}
	Given that $\hat{f}(\xi) = e^{-t|\xi|^{2\alpha}}$, we have $\hat{f}(0) = 1$. Consequently, in this low-decay regime, the amplitude amplification overhead is minimal and strictly constant:
	\begin{equation}
		\|f\|_{L^1} = 1 \sim \mathcal{O}(1).
	\end{equation}

	\section{Proof of Theorem~\ref{thm:even_alpha_complexity}}\label{Proof1}
	Let $\alpha \in \mathbb{Z}^+$ be an integer, assume we have access to the block encoding of $2\tilde{H}-I$ or $\sqrt{\tilde{H}}$, $H^\alpha=A$ and $H\succeq 0 $. The non-unitary dynamics of $u_T=e^{-T H^{\alpha}} u_0= e^{-T A} u_0$ can be solved by a quantum algorithm with precision $\epsilon$ to get normalized final state $\ket{u_T}$ with $\Omega(1)$ success probability. The query complexity to the matrix oracle for $2\tilde{H}-I$ is
	\begin{equation}
		\begin{split}
			\mathcal{O}\Bigg( u_r \log(1+\alpha) \bigg[ \|A\|^{\frac{1}{2\alpha}} T^{\frac{1}{2\alpha}} \left( \log\frac{u_r}{\epsilon} \right)^{1 - \frac{1}{2\alpha}}  + \log\frac{u_r}{\epsilon} \bigg] \Bigg),
		\end{split}
	\end{equation}
	the query complexity to the oracle for initial state preparation is $\Ocal(u_r\log(1+\alpha))$ and the number of LCU coefficients scales as:
	\begin{equation}
		\mathcal{O}\left( \|A\|^{\frac{1}{2\alpha}} T^{\frac{1}{2\alpha}} \left( \log\frac{u_r}{\epsilon} \right)^{1 - \frac{1}{2\alpha}} + \log\frac{u_r}{\epsilon} \right),
	\end{equation}  here $u_r := \|u_0\|/\|u_T\|$.
	
	\begin{proof}
		
		We consider the standard Fourier integral $f(x) = \int_{-\infty}^{\infty} e^{2\pi i x\xi - T\xi^{2\alpha}} d\xi$. For large $|x|$, the integral is dominated by the saddle point of the phase function $\Phi(\xi) = 2\pi i x\xi - T\xi^{2\alpha}$. To determine the decay envelope, we analyze the stationary point of the real magnitude function $g(\xi) = 2\pi x\xi - T\xi^{2\alpha}$. Solving $g'(\xi) = 2\pi x - 2T\alpha\xi^{{2\alpha}-1} = 0$ yields the saddle point:
		\begin{equation}
			\xi_* = \left( \frac{\pi x}{T\alpha} \right)^{\frac{1}{2\alpha-1}}.
		\end{equation}
		Substituting $\xi_*$ back into $g(\xi)$, the maximum exponent becomes:
		\begin{equation}
			g(\xi_*) = 2\pi x \xi_* - T \xi_*^{2\alpha} = ({2\alpha}-1) \left( \frac{\pi x}{\alpha  T^{1/(2\alpha)}} \right)^{\frac{2\alpha}{2\alpha-1}}.
		\end{equation}
		
		Consequently, the asymptotic behavior of the inverse Fourier transform is governed by the exponential envelope:
		\begin{equation}
			|f(x)| \sim \exp\left[ - (2\alpha-1) \left( \frac{\pi |x|}{\alpha \, T^{1/(2\alpha)}} \right)^{\frac{2\alpha}{2\alpha-1}} \right].
		\end{equation}
		
		This confirms the exponential decay stated in Eq.~\eqref{eq:poisson_operator_2}. The simulation accuracy is constrained by two error sources: domain truncation and aliasing. 
		
		We set the budget for both errors to $\epsilon'/2$. Based on the decay envelope, we require
		\begin{equation}
			\frac{1}{a} \sum_{|k| > K} f(k/a) e^{-i \frac{2\pi k}{a} \sqrt{H}} \le C\int_{K/a}^{+\infty}|f(x)|dx\leq \epsilon'/2,
		\end{equation}
		here $C$ is constant. To rigorously determine the truncation radius $K/a$ that bounds the tail integral $\mathcal{E}(K/a) = \int_{K/a}^{\infty} |f(x)| dx$ by a precision $\epsilon'/(2C)$, we start with the asymptotic decay envelope $|f(x)| \sim \exp[-\lambda x^\beta]$, where $\beta = \frac{2\alpha}{2\alpha-1}$ and $\lambda = (2\alpha-1) (\frac{\pi}{\alpha T^{1/(2\alpha)}})^\beta$. Since $\beta > 1$, the integral lacks a closed-form solution, but we can derive a strict upper bound using the identity $e^{-\lambda x^\beta} = -(\lambda \beta x^{\beta-1})^{-1} \frac{d}{dx}(e^{-\lambda x^\beta})$. Exploiting the monotonicity of the pre-factor $x^{-(\beta-1)}$ for $x \ge K/a$, we obtain the bound:
		\begin{equation}
			\mathcal{E}(K/a) < \int_{K/a}^{\infty} \frac{1}{\lambda \beta (K/a)^{\beta-1}} d\left( -e^{-\lambda x^\beta} \right) = \frac{e^{-\lambda (K/a)^\beta}}{\lambda \beta (K/a)^{\beta-1}}.
		\end{equation}
		Requiring this upper bound to be at most $\epsilon$, we take the logarithm to find $\lambda (K/a)^\beta + (\beta-1)\log (K/a) + \log(\lambda\beta) \ge \log(2C/\epsilon')$. By adopting a conservative estimation strategy that neglects the sub-dominant logarithmic terms $(\beta-1)\log (K/a)$ and $\log(\lambda\beta)$ as $(K/a)$ becomes large, the condition simplifies to $\lambda (K/a)^\beta \ge \log(2C/\epsilon')$. Solving for $(K/a)$ and substituting the physical parameters $\lambda$ and $\beta$ back into the expression yields the sufficient truncation condition:
		\begin{equation}
			K/a \ge \left( \frac{1}{\lambda} \log\frac{2C}{\epsilon'} \right)^{1/\beta} = \frac{\alpha}{\pi} \, t^{\frac{1}{2\alpha}} \left( \frac{\log(2C/\epsilon')}{2\alpha-1} \right)^{1 - \frac{1}{2\alpha}}.
		\end{equation}
		This formula confirms that $K$ scales linearly with $t^{\frac{1}{2\alpha}}$ and quasi-linearly with $\log(1/\epsilon')$, providing a computationally safe cutoff for the LCU expansion. Then we have
		\begin{equation}\label{eq:K/aeven}
			K/a\sim \Ocal\left( T^{\frac{1}{2\alpha}} \left( {\log\frac{1}{\epsilon'}} \right)^{1 - \frac{1}{2\alpha}}
			\right).
		\end{equation}
		
		The aliasing error is bounded by the sum of shifted spectral copies. Dominant contribution arises from the nearest neighbors ($n=\pm 1$):
		\begin{equation}
			\left\| \sum_{n \neq 0} e^{-T(\sqrt{H} + naI)^{2\alpha}} \right\| \le 2 \sum_{n=1}^{\infty} e^{-T n^{2\alpha} (a - \|\sqrt{H}\|)^{2\alpha}} \approx 2 e^{-T(a - \|\sqrt{H}\|)^{2\alpha}}.
		\end{equation}
		Ensuring $2 e^{-T(a - \|\sqrt{H}\|)^{2\alpha}} \le \epsilon'/2$ leads to the condition $a \sim\Ocal\left(\|\sqrt{H}\| + \left( \frac{1}{T} \log\frac{4}{\epsilon'} \right)^{\frac{1}{2\alpha}}\right) $.
		
		The LCU procedure requires amplitude amplification proportional to the $L^1$-norm of the coefficients, $\|c\|_1 \approx \int |f(x)| dx$. As $\alpha$ increases, $\hat{f}(\xi)$ approaches a rectangular function, causing $f(x)$ to exhibit Sinc-like behavior ($\sin(x)/x$). Since $\int |x|^{-1} dx$ diverges logarithmically, the norm scales as $\|c\|_1 = \mathcal{O}(\log \alpha)$ as Theorem~\ref{thm:L1} in Appendix~\ref{App:EstimL1}.
		
		To achieve a final success probability $\Omega(1)$, we employ amplitude amplification. The required number of steps is proportional to the amplification factor $u_r \|c\|_1 \approx u_r \log \alpha$. 
		
		The number of LCU coefficients (Select operator complexity) scales linearly with the truncation index $K$:
		\begin{equation}
			K\sim \frac{K}{a} \cdot a  \sim\mathcal{O}\left( \|A\|^{\frac{1}{2\alpha}} T^{\frac{1}{2\alpha}} \left( \log\frac{1}{\epsilon'} \right)^{1 - \frac{1}{2\alpha}} + \log\frac{1}{\epsilon'} \right).
		\end{equation}
		
		To analyze the quantum query complexity of implementing the target operators within the LCU framework, we consider the maximum required phase evaluated at the truncation radius $K$. Based directly on the results established in Lemma~\ref{lemma:qsvt_complexity}, achieving an approximation error bounded by $\epsilon'$ for the operator $\cos(\frac{K}{a}\sqrt{H})$ yields a total quantum query complexity of:
		\begin{equation}
			\mathcal{O}\left(\frac{K}{a}\sqrt{\|H\|} + \log\left(\frac{1}{\epsilon'}\right)\right)
		\end{equation}
		calls to the respective block-encoding oracle.
		
		Consequently, the internal precision must be adjusted to $\epsilon' = \frac{\epsilon}{u_r}$. We get the final query to $H$ for block-encoding is about
		\begin{equation}
			\mathcal{O}\left( u_r \log(1+\alpha) \cdot \left[ \|A\|^{\frac{1}{2\alpha}} T^{\frac{1}{2\alpha}} \left( \log\frac{u_r}{\epsilon} \right)^{1 - \frac{1}{2\alpha}} + \log\frac{u_r}{\epsilon} \right] \right)
		\end{equation} 
		and the total query complexity to the initial state preparation oracle scales as $O(u_r \log(1+\alpha))$, as the oracle must be queried in each of the amplitude amplification iterations.
		
		It is worth noting that Corollary~\ref{cor:even_alpha} explicitly dispenses with the assumption of positive semi-definiteness ($H \succeq 0$) for the root operator $H$. This relaxation is mathematically justified because our derivation directly evaluates the unitary evolution operators of the form $e^{-i\frac{2\pi k}{a}H}$, as formulated in Eq.~\eqref{eq:poisson_operator}. Within the QSVT framework, implementing such complex exponential functions solely requires the generator $H$ to be a Hermitian matrix, thereby entirely circumventing any positivity constraints on its spectrum.
	\end{proof}

	\section{Proof of Theorem~\ref{thm:fractional_complexity}}\label{Proof3}
	
	Let $\alpha \in \mathbb{R}^+ \setminus \mathbb{Z}$ and  be a non-integer. The non-unitary dynamics can be solved by a quantum algorithm with precision $\epsilon$ with $\Omega(1)$ success probability from Eq.~\eqref{eq:origin_problem}, assume we have access to the block encoding of $2\tilde{H}-I$ or $\sqrt{\tilde{H}}$ and $H\succeq 0$. Due to the fractional singularity, the query complexity to the matrix oracle for $H$ scales polynomially with the inverse precision:
	\begin{equation} 
		\mathcal{O}\left( u_r \log(\alpha+e) \alpha \cdot  \|A\|^{\frac{1}{2\alpha}} T^{\frac{1}{2\alpha}} \left( \frac{u_r}{\epsilon} \right)^{\frac{1}{2\alpha}} \right),
	\end{equation}
	where the dominant scaling is $\mathcal{O}(\epsilon^{-\frac{1}{2\alpha}})$. The number of LCU coefficients scales as:
	\begin{equation} 
		\mathcal{O}\left(\alpha \left(\frac{ u_r}{\epsilon}\right)^{\frac{1}{2\alpha}} \left(T^{\frac{1}{2\alpha}}\|A\|^{\frac{1}{2\alpha}}+(\log\frac{u_r}{\epsilon})^{\frac{1}{2\alpha}}\right)\right).
	\end{equation}
	And query complexity to $|u_0\rangle$ is about $\Ocal(u_r\log(e+\alpha))$, here $ u_r := \|u_0\|/\|u_T\|$.
	
	\begin{proof}
		The asymptotic behavior of the time-domain evolution kernel $f(x)$ is universally governed by the leading non-analytic term $|\xi|^{2\alpha}$ in the exponent. According to generalized Tauberian theorems, this fractional singularity dictates that the inverse Fourier transform exhibits an algebraic decay asymptotically scaling as:
		\begin{equation} \label{eq:fractional_decay}
			|f(x)| \sim \frac{C_{\alpha, T}}{|x|^{2\alpha+1}}, \quad \text{as } |x| \to \infty,
		\end{equation}
		where $C_{\alpha, T} \approx \frac{T}{\pi} \Gamma(2\alpha+1) \sin\left({\pi \alpha}\right)$~\cite{f2a4fe49-665d-3843-bb39-6a551e0f0422}. This heavy-tailed algebraic decay holds for all non-integer $\alpha$, intrinsically distinguishing this regime from the exponential convergence observed in integer orders.
		
		To determine the truncation radius $K/a$, we bound the spatial truncation error $\mathcal{E}$ by the internal precision threshold $\epsilon'$. Integrating the tail asymptotic from Eq.~\eqref{eq:fractional_decay} yields:
		\begin{equation}
			\mathcal{E} \approx 2 \int_{K/a}^{\infty} \frac{C_{\alpha, T}}{x^{2\alpha+1}} dx = 2 C_{\alpha, T} \left[ \frac{x^{-2\alpha}}{-2\alpha} \right]_{K/a}^{\infty} = \frac{C_{\alpha, T}}{\alpha (K/a)^{2\alpha}}.
		\end{equation}
		Setting $\mathcal{E} \le \epsilon'/2$ and solving for $K/a$, we derive the necessary scaling law for the truncation radius:
		\begin{equation} \label{eq:K_fractional_final}
			K/a \ge \left( \frac{2 C_{\alpha, T}}{\alpha \epsilon'} \right)^{\frac{1}{2\alpha}} \sim \mathcal{O}\left(\alpha \left(\frac{T\sin(\pi\alpha)}{\epsilon'}\right)^{\frac{1}{2\alpha}} \right).
		\end{equation}
		This result reveals the fundamental computational bottleneck of fractional dynamics: the truncation radius scales polynomially with the inverse precision $\mathcal{O}(\epsilon'^{-\frac{1}{2\alpha}})$, which is significantly more expensive than the quasi-polylogarithmic scaling achieved in analytic regimes.
		
		For the aliasing error, we analyze the spectral gap $D = a - \|\sqrt{H}\|$. The error is predominantly bounded by the nearest neighbor spectral copies ($n = \pm 1$):
		\begin{equation}
			\left\| \sum_{n \neq 0} e^{-T\| \sqrt{H} + naI \|^{2\alpha}} \right\| \le 2 e^{-T D^{2\alpha}} + 2 \sum_{n=2}^{\infty} e^{-T (nD)^{2\alpha}}.
		\end{equation}
		Using the integral bound $\int_1^{\infty} e^{-T D^{2\alpha} x^{2\alpha}} dx$ and the asymptotic expansion of the incomplete Gamma function $\Gamma(1/(2\alpha), T D^{2\alpha})$~\cite{NIST:DLMF}, we approximate the tail sum as $\frac{e^{-T D^{2\alpha}}}{2\alpha T D^{2\alpha}}$. The sufficient condition to bound the total aliasing error by $\epsilon'/2$ is given by:
		\begin{equation}
			C_{\alpha} e^{-T (a-\|\sqrt{H}\|)^{2\alpha}} \le \epsilon'/2 \implies a \sim \mathcal{O}\left(\|\sqrt{H}\| + \left( \frac{1}{T}\log \frac{1}{\epsilon'} \right)^{\frac{1}{2\alpha}} \right),
		\end{equation}
		where $C_{\alpha} \approx 2\left(1 + \frac{1}{2\alpha \log(1/\epsilon')}\right)$.
		
		Finally, we combine the truncation radius and the sampling rate, accounting for the amplitude amplification overhead derived in Theorem~\ref{thm:L1} of Appendix~\ref{App:EstimL1}. To achieve a final success probability of $\Omega(1)$, we must adjust the internal precision to $\epsilon' = \frac{\epsilon}{u_r}$. Based directly on the results established in Lemma~\ref{lemma:qsvt_complexity}, implementing the target operator $\cos(\frac{2\pi K}{a}\sqrt{H})$ to precision $\epsilon'$ requires a quantum query complexity of $\mathcal{O}(\frac{K}{a}\sqrt{\|H\|} + \log(\frac{1}{\epsilon'}))$ calls to the respective block-encoding oracle. Substituting the asymptotic bound for $K/a$, recognizing that $\sqrt{\|H\|} = \|A\|^{\frac{1}{2\alpha}}$, and multiplying by the $\mathcal{O}(u_r \log(\alpha+e))$ amplitude amplification iterations, we obtain the final query complexity to the block-encoding oracle:
		\begin{equation} 
			\mathcal{O}\left( u_r \log(\alpha+e) \alpha \cdot  \|A\|^{\frac{1}{2\alpha}} T^{\frac{1}{2\alpha}} \left( \frac{u_r}{\epsilon} \right)^{\frac{1}{2\alpha}} \right).
		\end{equation}
		Correspondingly, the initial state preparation requires $\mathcal{O}(u_r \log(\alpha+e))$ queries. The total number of LCU coefficients is bounded by $K = (K/a) \cdot a$, which is dominated by the polynomial term from $K/a$, yielding:
		\begin{equation} 
			\mathcal{O}\left(\alpha \left(\frac{ u_r}{\epsilon}\right)^{\frac{1}{2\alpha}} \left(T^{\frac{1}{2\alpha}}\|A\|^{\frac{1}{2\alpha}}+\left(\log\frac{u_r}{\epsilon}\right)^{\frac{1}{2\alpha}}\right)\right).
		\end{equation}
	\end{proof}
\end{document}